\documentclass[reqno]{amsart}
\usepackage{amsmath}
\usepackage{amsfonts}
\usepackage{amssymb}
\usepackage{amsgen}
\usepackage{mathrsfs}
\usepackage{cite}
\usepackage[all]{xy}
\usepackage[active]{srcltx}

\newcommand{\inv}{^{-1}}
\newcommand{\p}{\varphi}

\newcommand{\ov}[1]{\ensuremath{\overline {#1}}}

\newcommand{\Thmname}{Theorem}
\newcommand{\Propname}{Proposition}
\newcommand{\Lemmaname}{Lemma}
\newcommand{\Definitionname}{Definition}
\newtheorem{Thm}{\Thmname}[section]
\newtheorem{Prop}[Thm]{\Propname}
\newtheorem{Lemma}[Thm]{\Lemmaname}
{\theoremstyle{definition}
}
{\theoremstyle{remark}
}

{\theoremstyle{remark}
\newtheorem{Example}[Thm]{Example}}

\newtheorem*{Lemma*}{Lemma}

\numberwithin{equation}{section}

\author{Benjamin Steinberg}
\address{Department of Mathematics \\ City College of New York}
\email{bsteinberg@ccny.cuny.edu}
\title{Topological dynamics and recognition of languages}
\date{March 23, 2010}
\thanks{This was written while the author was at the School of Mathematics of Carleton University.  The author was supported at the time by an NSERC grant}

\begin{document}

\begin{abstract}
We define compact automata and show that every language has a unique minimal compact automaton.  We also define recognition of languages by compact left semitopological monoids and construct the analogue of the syntactic monoid in this context.  For rational languages this reduces to the usual theory of finite automata and finite monoids.
\end{abstract}
\maketitle
\section{Introduction}
The theory of automata and syntactic monoids is quite successful for rational languages~\cite{EilenbergA,Eilenberg,qtheor,Almeida:book} because the algebraic invariants are finite automata, finite monoids and finite transformation monoids, which are objects with a fair amount of structure.  On the other hand, the corresponding theory for non-rational languages has been far less successful, in a great part due to the fact that infinite semigroups and transformation semigroups have little structure.  In particular, one may have no non-identity idempotents and hence no maximal subgroups.

The idea here is to replace the minimal automaton of a language with a topological dynamical system and the syntactic monoid with a compact left semitopological monoid.  For the case of a rational language this reduces to the standard theory.  However, ever left semitopological compact monoid that is not a group has non-identity idempotents by a theorem of Ellis~\cite{Ellis}.  Thus we have a much richer structure.  The reader should consult~\cite{Hindman,Akin,Ellis} for more on left semitopological compact semigroups and in particular with respect to applications to dynamics and Ramsey theory.

It seems clear that one can then develop varieties of left semitopological compact monoids, define ``pseudo''identities using the Stone-\v{C}ech compactification of the free monoid and given an Eilenberg correspondence between varieties of languages and varieties of left semitopological compact monoids.  I hope someone will develop this theory and find applications for it, or perhaps at some future point I will enlarge this draft.

\section{Definitions}
Let $A$ be a finite set and $A^*$ be the corresponding free monoid.  By a \emph{topological dynamical system} over $A$ we mean a compact Hausdorff space $X$ equipped with a right action of $A^*$ by continuous functions.  We say the system is \emph{metric} if $X$ is metric.  Such a right action is uniquely determined by a continuous transition function $\delta\colon X\times A\to X$.   A \emph{morphism} of topological dynamical systems $(X,A,\delta)\to (Y,A,\lambda)$ is a continuous map $\p\colon X\to Y$ such that $\p(xa)=\p(x)a$ for all $x\in X$ and $a\in A$.

A \emph{topological automaton} $\mathscr A$ is a $5$-tuple $(X,A,x_0,\delta,T)$ where $(X,A,\delta)$ is a topological dynamical system, $x_0\in X$ and $T\subseteq X$ is a clopen subset.
Of course, we say a topological automaton is \emph{metric} if the underlying topological dynamical system is metric.
The language $L(\mathscr A)$ accepted by $\mathscr A$ is the set of all words $w\in A^*$ such that $x_0w\in T$. We say that $\mathscr A$ is \emph{trim} if $x_0A^*$ is dense in $X$.  Clearly, replacing $X$ by $\ov{x_0A^*}$ and $T$ by $\ov{x_0A^*}\cap T$ results in a trim topological automaton accepting the same language.

\begin{Prop}\label{existenceofrecognizer}
Let $L\subseteq A^*$.  Let $\chi_L$ be the characteristic function of $L$ and let $X=\{0,1\}^{A^*}$ equipped with the product topology.  Define a transition function $\delta\colon X\times A\to X$ by $\delta(f,a)(w)=f(aw)$. Let $T=\{f\in X\mid f(\varepsilon)=1\}$.  Then $L$ is accepted by the metric topological automaton $(X,A,\chi_L,\delta,T)$.
\end{Prop}
\begin{proof}
The set $T$ is clearly clopen.
Trivially, $w\in L$ if and only if $\chi_L(w)=1$ if and only if $(\chi_Lw)(\varepsilon)=1$ if and only if $\chi_Lw\in T$.
\end{proof}

We define the \emph{minimal topological automaton} of $L$ to be \[\mathscr A_L = (\ov{\chi_LA^*},A,\chi_L,\delta, T\cap \ov{\chi_LA^*}).\]  It is the closure of the state set $\chi_LA^*$ of the classical minimal automaton in $\{0,1\}^{A^*}$.  In particular, if $L$ is rational, then $\mathscr A_L$ is finite and is the usual minimal automaton.  Let us prove that $\mathscr A_L$ is truly minimal.  For convenience of notation, put $\mathscr O_L = \ov{\chi_L A^*}$ and $T_L = \{f\in \mathscr O_L\mid f(\varepsilon)=1\}$.

\begin{Thm}\label{minimization}
Let $L\subseteq A^*$ be a language and suppose $\mathscr A=(Y,A,y_0,\lambda,F)$ is a trim topological automaton accepting $L$.  Then there is a unique surjective morphism $\p\colon (Y,A,\lambda)\to (\mathscr O_L,A,\delta)$ of topological dynamical systems over $A$ such that $\p(y_0)=\chi_L$.  Moreover, $\p(F)=T_L$.
\end{Thm}
\begin{proof}
If $\p$ exists, it is evidentally unique and surjective by trimness.  Let us prove existence; we continue to put $X=\{0,1\}^{A^*}$. There is a natural left action of $A^*$ on the space $\{0,1\}^Y$ given by $uf(y) = f(yu)$.
For $y\in Y$, define $\p_y\colon A^*\to \{0,1\}$ by $\p_y(u) = u\chi_F(y)$.  Since $T$ is clopen, $\chi_F$ is continuous and hence one easily checks that $\p\colon Y\to X$ given by $\p(y)=\p_y$ is continuous for the product topology on $X$ (using that multiplication by $u$ is continuous).  Now $\p_{y_0}(u) = u\chi_F(y_0) = \chi_F(y_0u) = \chi_L(u)$ because $\mathscr A$ accepts $L$.  It now follows that $\p$ takes $Y$ into $\mathscr O_L$ by the trimness of $\mathscr A$.  Also note that if $y\in F$, then $\p_y(\varepsilon) = \chi_F(y) = 1$ and so $\p(y)\in T_L$.  Conversely, if $t\in T_L$ and $t=\p(y)$, then $\chi_F(y) = \p_y(\varepsilon) = t(\varepsilon)=1$.  Thus $y\in F$ and so $t\in \p(F)$.  This completes the proof.
\end{proof}

Next we want to define the enveloping monoid of a topological dynamical system $(X,A,\delta)$.  Let $S=X^X$ where we let elements of $S$ act on the \emph{right} of $X$.  We topologize $S$ with the product topology and hence it is compact Hausdorff.  In fact, $S$ is a compact left semitopological semigroup, that is, a non-empty, compact Hausdorff space with a semigroup structure such that, for each $s\in S$, the left translation $t\mapsto st$ is continuous.  Indeed, if $f_{\alpha}\to f$ is a net, then $xgf_{\alpha}\to xgf$ precisely because we are using the topology of pointwise convergence. The \emph{topological center} of a compact left semitopological semigroup $T$ consists of all elements $s\in T$ for which the right translation $t\mapsto ts$ is also continuous.  It is a submonoid of $S$ containing the algebraic center.  The topological center of $S=X^X$ is the submonoid of continuous functions.  The theory of compact left semitopological semigroups is quite well developed, see for instance~\cite{Hindman,Akin}.

\begin{Prop}\label{canclose}
Let $S$ be a compact left semitopological monoid and suppose $T$ is a submonoid of the topological center of $S$.  Then $\ov T$ is a closed submonoid and hence compact left semitopological in its own right.
\end{Prop}
\begin{proof}
Suppose that $s,t\in \ov T$ and $s_{\alpha}\to s$ and $t_{\beta}\to t$ are nets in $T$.  Then $st = \lim_{\beta} st_{\beta}$ by continuity of left translation by $s$.  Hence $st=\lim_{\beta} \lim_{\alpha}s_{\alpha}t_{\beta}$ by continuity of right translation by the $t_{\beta}$ (using that $T$ is contained in the topological center).  Thus $st\in \ov T$ and so $\ov T$ is a closed submonoid.
\end{proof}

Let $(X,A,\delta)$ be a topological dynamical system over $A$ and let $\delta^*\colon A^*\to X^X$ be the induced map.  Then the \emph{transition monoid} of $\mathscr A=(X,A,\delta)$ is $M(\mathscr A)=\delta^*(A^*)$ and the \emph{enveloping monoid} is $E(\mathscr A)=\ov{M(\mathscr A)}$, which is a compact left semitopological monoid by Proposition~\ref{canclose}.  Moreover, the action $X\times E(\mathscr A)\to X$ is an \emph{Ellis action}~\cite{Akin}, meaning that, for each $x\in X$, the $s\mapsto xs$ is continuous (actually in~\cite{Akin} left actions are considered, but this is immaterial).

Now if we consider the minimal topological automaton $\mathscr A_L$ of $L\subseteq A^*$, then $M(\mathscr A_L)$ is the usual syntactic monoid of $L$, written $M_L$ for short.  We define $E(\mathscr A_L)$ to be the \emph{enveloping syntactic monoid of $L$} and denote it $E_L$ for brevity.  If $L$ is rational, then $E_L$ is the usual syntactic monoid of $L$.

Our goal is to show that $E_L$ is minimal in the usual sense among compact left semitopological monoids recognizing $L$.   Let us say that a homomorphism $\p\colon A^*\to M$ with $M$ a compact left semitopological monoid \emph{recognizes} $L$ if $\p(A^*)$ is contained in the topological center of $M$ and there is a clopen subset $F$ of $M$ such that $L = \p\inv(F)$.  Of course, one can then replace $M$ by $\ov{\p(A^*)}$ and $F$ by $F\cap \ov{\p(A^*)}$ thanks to Proposition~\ref{canclose}.

\begin{Prop}\label{clopenfinalelements}
Let $\mathscr A=(X,A,x_0,\delta,T)$ be a topological automaton accepting $L\subseteq A^*$.  Then the natural morphism $\delta^*\colon A^*\to E(\mathscr A)$ recognizes $L$.
\end{Prop}
\begin{proof}
First of all $\delta^*(A^*)$ is contained in the topological center of $E(\mathscr A)$ since each of its elements is a continuous function on $X$.
Let \[F=\{s\in E(\mathscr A)\mid x_0s\in T\}.\]  Clearly one has $(\delta^*)\inv (F)=L$. It remains to check that $F$ is clopen.  But if $f_{x_0}\colon S\to X$ is the map $f_{x_0}(s) = x_0s$, then $f_{x_0}$ is continuous and $F= f_{x_0}^{-1}(T)$, and hence is clopen.
\end{proof}

In particular, if $\eta\colon A^*\to M_L$ is the syntactic morphism, then $\eta\colon A^*\to E_L$ (abusing notation) recognizes $L$ via the set $F_L = \{s\in M_L\mid \chi_Ls(\varepsilon)=1\}$.

\begin{Lemma}\label{functorial}
Let $\mathscr A=(X,A,\delta)$ and $\mathscr B=(Y,A,\lambda)$ be topological dynamical systems.
Suppose that $\p\colon \mathscr A\to \mathscr B$ is a surjective continuous morphism of topological dynamical systems.  Then $\p$ induces a unique surjective continuous homomorphism $\psi\colon E(\mathscr A)\to E(\mathscr B)$ of enveloping monoids such that \[\xymatrix{A^*\ar[rr]^{\delta^*}\ar[rd]_{\lambda^*} & & E(\mathscr A)\ar[ld]^{\psi}\\ & E(\mathscr B)& }\] commutes.
\end{Lemma}
\begin{proof}
Uniqueness and surjectivity are immediate from the density of the transition monoid in the enveloping monoid. Define $\psi\colon E(\mathscr A)\to Y^Y$ by $y\psi(s) = \p(xs)$ where $x$ is any $\p$-preimage of $y$.  First we show that $\psi$ is well defined.  Indeed, suppose $\p(x')=y$.  Let $w_{\alpha}$ be a net from $A^*$ with $\delta^*(w_{\alpha})\to s$.  Then since $t\mapsto xt$ is continuous, we have $xw_{\alpha}\to xs$ and $x'w_{\alpha}\to x's$.  Then $\p(xs)=\lim \p(xw_{\alpha}) = \lim yw_{\alpha}$ and similarly $\p(x's)=\lim yw_{\alpha}$ showing that $y\psi(s)$ is well defined.  Suppose $y=\p(x)$.  Then $y\psi(s) = \p(xs)$ and so $(y\psi(s))\psi(t) = \p((xs)t) = \p(x(st))= y\psi(st)$.  It follows that $\psi$ is a homomorphism. Let us show that
$\psi$ is continuous.  Suppose $s_{\alpha}$ is a net converging to $s$ and fix $y\in Y$.  Choose $x$ with $\p(x)=y$.  Then $xs_{\alpha}\to xs$ and so $y\psi(s_{\alpha}) = \p(ys_{\alpha})\to \p(ys) = y\psi(s)$, establishing the continuity of $\psi$.   Since $\psi(\delta^*(A^*)) = \lambda^*(A^*)$, we conclude that $\psi(E(\mathscr A))= E(\mathscr B)$, completing the proof.
\end{proof}

We are now ready to establish the desired minimality property of $E_L$.

\begin{Thm}\label{minimalmonoid}
Let $\p\colon A^*\to M$ be a homomorphism recognizing $L\subseteq A^*$ by a clopen subset $F$ and suppose that $\p(A^*)$ is dense.  Then there is a unique continuous surjective homomorphism $\psi\colon M\to E_L$ such that
\[\xymatrix{A^*\ar[rr]^{\p}\ar[rd]_{\eta} & & M\ar[ld]^{\psi}\\ & E_L& }\] commutes.
\end{Thm}
\begin{proof}
Uniqueness is clear from density.  For existence, consider the topological automaton $\mathscr A=(M,A,1_M,\mu,F)$ where $\mu(m,a) = m\p(a)$.  The fact that $\p(a)$ is in the topological center says exactly that $m\mapsto m\p(a)$ is continuous. It is straightforward to verify $E(\mathscr A)=M$ acting via right multiplication.  Clearly $\mathscr A$ recognizes $L$.   Proposition~\ref{minimization} then provides a continuous surjective morphism $\rho\colon (M,A,\mu)\to (\mathscr O_L,A,\delta)$.  The theorem now follows from Lemma~\ref{functorial}.
\end{proof}

\begin{Example}
Let $A$ have one symbol, and so we can identify $A^*$ with $\mathbb N$.  Then $\{0,1\}^{\mathbb N}$ is the usual space of infinite words and the generator of $\mathbb N$ acts by the shift map removing the first letter.  Suppose $L\subseteq \mathbb N$ is a language such that $\chi_L$ has a dense orbit.  For instance, if $w_0,w_1,\ldots$ is the length-lexicographic enumeration of $\{0,1\}^*$, then the language whose characteristic sequence is the word $w_0w_1w_2\cdots$ has a dense orbit.  Thus $\mathscr O_L =\{0,1\}^{\mathbb N}$.  It is well known that the enveloping monoid of the shift map on $\{0,1\}^{\mathbb N}$ is isomorphic to the Stone-\v{C}ech compactification $\beta \mathbb N$ of $\mathbb N$~\cite{Hindman}.  This shows that one cannot recognize every language by a metrizable compact left semitopological monoid.
\end{Example}

Let us observe that if one uses compact topological monoids (where multiplication is jointly continuous), then one obtains nothing beyond rational languages and so the semitopological nature of the recognizing monoid is crucial.

\begin{Prop}
Suppose $\p\colon A^*\to M$ is a homomorphism recognizing $L$ with $M$ a compact topological monoid.  Then $L$ is rational.
\end{Prop}
\begin{proof}
Let $F$ be a clopen subset of $M$ such that $\p\inv (F)=L$.  It is well known that the syntactic congruence of a clopen subset of a compact topological monoid $M$ is open (as a subset of $M\times M$) (c.f.~\cite{qtheor}).  Hence there is a continuous homomorphism $\psi\colon M\to N$ with $N$ a finite monoid such that $\psi\inv \psi(F) =F$.  Then $(\psi\p)\inv (\psi(F)) = L$ and so $L$ is rational.
\end{proof}

It is interesting to ask for which languages $L\subseteq A^*$ is the dynamical system $(\mathscr O_L,A,\delta)$ minimal. In the case that $L$ is rational, this corresponds to the minimal automaton for $L$ being strongly connected.

\def\malce{\mathbin{\hbox{$\bigcirc$\rlap{\kern-7.75pt\raise0,50pt\hbox{${\tt
  m}$}}}}}\def\cprime{$'$} \def\cprime{$'$} \def\cprime{$'$} \def\cprime{$'$}
  \def\cprime{$'$} \def\cprime{$'$} \def\cprime{$'$} \def\cprime{$'$}
  \def\cprime{$'$}


\end{document}